\documentclass[11pt,a4paper]{article}
\usepackage[cp1251]{inputenc}
    \usepackage[english]{babel}
\usepackage{algorithm}
\usepackage[noend]{algpseudocode}    
    \usepackage{amsmath,amsthm,amssymb}

\usepackage{tikz}
\usepackage{xypic}
\usetikzlibrary{shapes}
\usetikzlibrary{arrows.meta}
     \usepackage{graphicx}
     \usepackage{wrapfig}
  \DeclareMathOperator\CSP{CSP}

  \DeclareMathOperator\Pol{Pol}

  \DeclareMathOperator\QCSP{QCSP}

\renewcommand{\le}{\leqslant}

\theoremstyle{definition}
\theoremstyle{plain}
\newtheorem{thm}{Theorem}

\newtheorem{lem}[thm]{Lemma}

\newtheorem{cor}[thm]{Corollary}

\title{The complete classification for quantified equality constraints}
\author{Dmitriy Zhuk, Barnaby Martin and Micha\l{} Wrona\thanks{Jagiellonian University, Krak\'ow. This author is partially supported by National Science Centre, Poland grant number 2020/37/B/ST6/01179.}}

\begin{document}

\maketitle

\begin{abstract}
    We prove that QCSP$(\mathbb{N};x=y\rightarrow y=z)$ is PSpace-complete, settling a question open for more than ten years. This completes the complexity classification for the QCSP over equality languages as a trichotomy between Logspace, NP-complete and PSpace-complete.
    
    We additionally settle the classification for bounded alternation QCSP$(\Gamma)$, for $\Gamma$ an equality language. Such problems are either in Logspace, NP-complete, co-NP-complete or rise in complexity in the Polynomial Hierarchy. 
\end{abstract}

\section{Introduction}

The \emph{quantified constraint satisfaction problem}, QCSP$(\Delta)$, asks whether some input sentence, whose quantifier-free part is positive conjunctive, is true on the \emph{constraint language} $\Delta$ comprised of a set of relations over the same domain $D$. The QCSP is itself a generalisation of its better-studied brother, the \emph{constraint satisfaction problem}, CSP$(\Delta)$, in which one assumes the quantification of the sentence is entirely existential. By now, a great deal is known about the CSP, while many important questions for the QCSP remain open. 

The celebrated Feder-Vardi Dichotomy Conjecture, which asserts that CSP$(\Delta)$ is always either in P or is NP-complete, when $\Delta$ is over a finite domain, has been proved independently in \cite{BulatovFVConjecture} and \cite{ZhukFVConjecture,Zhuk20}. This provides the basis for some CSP classifications over an infinite domain \cite{BodirskyMM18,BodirskyM18}. In the meantime, the study of infinite-domain CSPs has itself matured to become a field of active research incorporating Universal Algebra and Model Theory. Important results from its formative years include \cite{BodirskyK10,BodirskyP15,BodirskyMM18} (as well as the earlier \cite{KrokhinJJ03}), but most of the post-modern work concerns the finitely-bounded $\omega$-categorical case \cite{BartoKOPP17,BartoP20}.

At present, we are far from a complexity classification for finite domain QCSPs. The $3$-element case with constants is settled in \cite{zhuk2020qcsp}, where evidence is also given that the classification is more complicated than previously thought. The study of infinite-domain QCSPs is also far less advanced than that for infinite-domain CSPs, with only a partial known classification in the temporal case (see \cite{CharatonikW08,CharatonikW08-bis,ChenW12,Wrona14}). 

The modern, systematic study of infinite-domain CSPs began with a complexity classification for equality languages, which are defined as those which have a first-order definition in $(\mathbb{N};=)$ \cite{BodirskyK08}. For such languages, the CSP is always in NP and the QCSP always in PSpace. It made sense to begin the systematic study of infinite-domain QCSPs similarly, with the equality languages, spawning the works \cite{BodirskyC07,BodirskyC10}. There was an error in the conference version \cite{BodirskyC07} which was corrected in the journal version \cite{BodirskyC10} at the cost of a trichotomy between Logspace, NP-complete and Co-NP-hard instead of  Logspace, NP-complete and PSpace-complete. Thus, a gap was left in this foundational result. It was known that there were PSpace-complete cases, e.g. QCSP$(\mathbb{N};x=y\rightarrow u=v)$, but it was not known whether QCSP$(\mathbb{N};x=y\rightarrow y=z)$ was Co-NP-complete or PSpace-complete (or even somewhere in between). Indeed, this constraint language became notorious; more so, because a solution of it as either Co-NP-complete or PSpace-complete would complete the classification.

In an invited talk at the \emph{Arbeitstagung Allgemeine Algebra} (Workshop on General Algebra -- AAA98) in July 2019 in Dresden, Hubie Chen posed the complexity of QCSP$(\mathbb{N};x=y\rightarrow y=z)$ as one of three open problems in the area of the complexity of constraints that apparently needed new insights to resolve \cite{ChenAAA2019}. In this paper we settle the complexity of QCSP$(\mathbb{N};x=y\rightarrow y=z)$ as PSpace-complete.

We further consider the bounded alternation restrictions of the QCSP, in which there is an a priori bound on the quantifier alternations. The corresponding problems are denoted $\Sigma_k$-QCSP$(\Gamma)$ and $\Pi_k$-QCSP$(\Gamma)$, when the input is restricted to be $\Sigma_k$ and $\Pi_k$, respectively. Such problems arise especially naturally in the 2-element (Boolean) case, for example propositional \emph{abduction} and \emph{prioritization} (see \cite{EiterG95}). The complexity classification for bounded alternation QCSP has been known to display differences from the unbounded case since \cite{Chen09}. Indeed, there are constraint languages $\Gamma_1$ and $\Gamma_2$, both on three elements, so that QCSP$(\Gamma_1)$ and QCSP$(\Gamma_2)$ are PSpace-complete, yet their bounded alternation versions are Co-NP-complete and rising in the polynomial hierarchy, respectively. 

It has been thought that bounded alternation QCSP$(\mathbb{N};x=y\rightarrow u=v)$ is Co-NP-complete. Indeed, Hubie Chen has been saying this for many years. The idea was that the methods of \cite{Chen09} could be adapted for $(\mathbb{N};x=y\rightarrow u=v)$. However, the direct application of \cite{Chen09} is not completely trivial in the infinite-domain case. Our proof does not use the machinery of \cite{Chen09}, though it has various similarities with it.

We prove a complete complexity classification for the bounded alternation QCSP$(\Gamma)$, along the lines that the authors of \cite{BodirskyC10} would have expected, where $\Gamma$ is an equality language. Each of the complexity classes -- Logspace, NP-complete, Co-NP-complete and rising in the polynomial hierarchy -- appear.

\subsection{The classifications}

Equality constraint languages admit quantifier elimination\footnote{Indeed, equality sits within all models, so must be dealt with in any discussion of quantifier elimination. For more details on this, we refer the reader to Section 2.7 in \cite{hodges_1993}.} so let us assume $\Gamma$ is a finite set of relations over an infinite countable set (e.g. $\mathbb{N}$) where each relation is defined by some Boolean combination of atoms of the form $x=y$. Let us additionally assume this Boolean combination is in conjunctive normal form (CNF). A relation is \emph{negative} if it has a CNF definition in which all of the clauses are either equalities, or are disjunctions of negative literals (i.e. of the form $x \neq y$). A relation is \emph{positive}  if it has a CNF definition in which all of the literals are positive (i.e. of the form $x = y$). A relation is \emph{Horn} if it has a CNF definition in which each clause contains at most one positive literal. $\Gamma$ is \emph{negative} (respectively, \emph{positive}, \emph{Horn}) if all of its relations are negative (respectively, positive, Horn). Let $I$ be the relation $x=y\rightarrow y=z$.

\begin{cor}
Let $\Gamma$ be an equality constraint language. Either
\begin{itemize}
    \item $\Gamma$ is negative, and QCSP$(\Gamma)$ is in Logspace, or else
    \item $\Gamma$ is positive, and QCSP$(\Gamma)$ is NP-complete, or else
    \item QCSP$(\Gamma)$ is PSpace-complete.
\end{itemize}
\end{cor}
\begin{proof}
We follow the argument as laid out in Theorem~5.5 of \cite{BodirskyC10}. If $\Gamma$ is negative then QCSP$(\Gamma)$ is in Logspace by Theorem 6.1 of \cite{BodirskyC10}. If $\Gamma$ is positive but not negative then QCSP$(\Gamma)$ is NP-complete by Theorem 7.1 of \cite{BodirskyC10}. Otherwise, by Theorem 8.1 of \cite{BodirskyC10}, $I$ is definable as a quantified conjunctive (positive) formula within $\Gamma$, and the result follows from our Corollary~\ref{cor:main}.
\end{proof}
Let us give examples of each of the three complexity classes in this classification. A non-trivial negative language might be $(x\neq y \vee u\neq v) \wedge v=w$. A positive language that is not negative might be $x=y \vee u=v$. Finally, the canonical example for the third class is, of course, $I$!

The classification for bounded alternation QCSP over equality languages is as follows.
\begin{thm}
Let $\Gamma$ be an equality language and $k \geq 2$.
\begin{itemize}
    \item If $\Gamma$ is negative, then $\Pi_k$-QCSP$(\Gamma)$ is in Logspace.
    \item Else, if $\Gamma$ is positive, then $\Pi_k$-QCSP$(\Gamma)$ is NP-complete.
    \item Else, if $\Gamma$ is Horn, then $\Pi_k$-QCSP$(\Gamma)$ is Co-NP-complete.
    \item Else, $\Pi_k$-QCSP$(\Gamma)$ is $\Pi_{k-2}^{\mathrm{P}}$-hard. 
\end{itemize}
\label{thm:bounded-alternation}
\end{thm}

\section{Preliminaries}

 We identify a constraint language $\Gamma$ with a set of relations over a fixed domain $D$, which we may always take in this paper to be $\mathbb{N}$. We may also think of this as a first-order relational structure. All $\Gamma$ in this paper will have a finite signature.

We always may assume that an instance of $\QCSP(\Gamma)$ is of the prenex form $$\forall x_1\exists y_1 \forall x_2 \exists y_2 \dots \forall x_{n} \exists y_{n} \Phi,$$ since if it is not it may readily be brought into such a form in polynomial time. Then a solution is a sequence of (Skolem) functions 
$f_{1},\ldots, f_{n}$ such that 
$$(x_1,f(x_1),x_2,f_2(x_1,x_2),\ldots,x_{n},f_n(x_1,\ldots,x_n))$$ is a solution of $\Phi$ for all 
$x_{1},\ldots,x_{n}$ (i.e. $y_{i} = f_{i}(x_{1},\ldots,x_{i})$). This belies a (Hintikka) game semantics for the truth of a QCSP instance in which a player called Universal (male) plays the universal variables and a player called Existential (female) plays the existential variables, one after another, from the outside in. The Skolem functions above give a strategy for Existential. In our proofs we may occasionally revert to a game-theoretical parlance.

A formula in conjunctive normal form with $3$-literals per clause is called a 3-CNF. We indicate negation of a propositional variable $z_i$ with the overline notation $\overline{z}_i$.

A key role in classifying CSP and QCSP has been played by Universal Algebra. 
We say that a $k$-ary operation $f$ \emph{preserves} 
an $m$-ary relation $R$, 
whenever $(x^1_1,\ldots,x^m_1),\ldots,(x^1_k,\ldots,x^m_k)$ in $R$, then also $(f(x^1_1,\ldots,x^1_k),$ $\ldots,f(x^m_1,\ldots,x^m_k))$ in $R$.
The relation $R$ is called \emph{an invariant} of $f$, 
and the operation $f$ is called \emph{a polymorphism} of $R$.
An operation $f$ is \emph{a polymorphism} of $\Gamma$ if it preserves every relation 
from $\Gamma$. The \emph{polymorphism clone} $\Pol(\Gamma)$ is the set of all polymorphisms of $\Gamma$. The restriction of this to surjective operations is denoted sPol$(\Gamma)$.

A formula of the form 
$\exists y_{1}\dots\exists y_{n} \Phi$,
where $\Phi$ is a  conjunction of relations from $\Gamma$ is called 
\emph{a positive primitive formula (pp-formula) over $\Gamma$}.
If $R(x_{1},\dots,x_{n}) = \exists y_{1}\dots\exists y_{n} \Phi$, 
then we say that $R$ is \emph{pp-defined} by 
$\exists y_{1}\dots\exists y_{n}$ $\Phi$, 
and $\exists y_{1}\dots\exists y_{n} \Phi$ is called 
\emph{a pp-definition}.
If we augment the definition of pp-formula with universal quantification, then we get the richer class of \emph{quantified conjunctive formulas} for which we can also define quantified conjunctive definitions. If $\Gamma$ is an equality language, it is known that the class of relations pp-definable on $\Gamma$ is precisely those relations that are preserved by the polymorphisms of $\Gamma$ \cite{BodirskyN06}. Furthermore, it is known that the class of relations quantified conjunctive definable on $\Gamma$ is precisely those relations that are preserved by the surjective polymorphisms of $\Gamma$ \cite{BodirskyC10}. This latter property is famously not yet known to hold for $\omega$-categorical structures in general \cite{ChenMuellerLMCS}.


A $k$-ary operation $f$ is \emph{essentially unary} if there exists $i \in [k]$ so that $f(x_1,\ldots,x_k)=g(x_i)$ for some unary function $g$. An operation that is not essentially unary is said to be \emph{essential}. 

\section{QCSP$(\mathbb{N};x=y\rightarrow y=z)$}



We will define 
formulas (constraints) over $I:=x=y \rightarrow y=z$ by diagrams.
An edge labeled with $z$ from a vertex $x$ to a vertex $y$ means
the constraint 
$I(x,z,y)$.
Sometimes we do not label a vertex. This means that the corresponding variable is existentially quantified after all other variables are quantified (innermost).
For instance, in Figure \ref{GraphExample} there is one unlabeled vertex. If we call it $u$ then 
Figure \ref{GraphExample} can be interpreted as the formula
$$\exists u (I(x,y,u)\wedge I(x,z,u)\wedge I(t,z,u)).$$

We define a reduction from 
the problem \emph{Quantified $3$-Satisfiability}. 
Suppose we have a formula $\Phi$, whose quantifier-free part is a propositional formula in 3-CNF, that is a formula of the form
\begin{align}\label{RunningExample}
\exists x_1\forall y_{1}\exists x_2\forall y_2\dots\exists x_n\forall y_n
\;((x_2\vee \overline y_1\vee y_3)
\wedge 
(\overline x_1\vee \overline x_2\vee x_5)\wedge\dots)
\end{align}
We define a formula $\Psi$ by
$$
\begin{array}{ll}
\forall t\forall f 
&
\exists x_{1}^{0}\forall x_{1}^{1}\;\;
\forall y_{1}^0\forall y_{1}^1 \;\;
\exists x_{2}^{0}\forall x_{2}^{1}\;\;
\forall y_{2}^0\forall y_{2}^1 \dots
\exists x_{n}^{0}\forall x_{n}^{1}\;\;
\forall y_{n}^0\forall y_{n}^1\;\; \\
& \exists z \;\;\exists \dots\exists 
(\mathcal C_0\wedge \mathcal C_{1}\wedge\dots\wedge \mathcal C_{n}\wedge \mathcal F),\\
\end{array}
$$
where 
\begin{itemize}
\item $\mathcal C_{0}$ is defined in Figure \ref{fig:C0}.
\item $\mathcal C_{1},\dots,\mathcal C_{n}$ are defined in Figure \ref{fig:Ci}; for the extreme cases 
    $\mathcal C_{n-1}$ and 
    $\mathcal C_{n}$ to be clear, they are also shown in Figures \ref{fig:Cn1} and \ref{fig:Cn}.

\item $\mathcal F$ encodes the 3-CNF and is defined as follows. Associated with each clause $(\ell_1 \vee \ell_2 \vee \ell_3)$ is a path of length $3$ from $t$ to $z$ with edge labels, in sequence, $\lambda(\ell_1),\lambda(\ell_2),\lambda(\ell_3)$, where $\lambda(\ell)=u^0$, iff $\ell$ is a positive variable $u$, and $\lambda(\ell)=u^1$, iff $\ell$ is a negative  variable $\overline{u}$.
In Figure \ref{fig:F} we show 
$\mathcal F$ for the formula (\ref{RunningExample}). 
\item $\exists \dots\exists$ means that all the remaining variables are existentially quantified.
\end{itemize}

Below we explain how to interpret different 
values of the variables:    
\begin{enumerate}
    \item variables $t$ and $f$ encode the value true and the value false;
    \item variables 
    $x_{i}^0$ and $x_{i}^1$ encode the variable $x_{i}$;
    $x_{i}^{0} = t$ means that $x_{i}=0$,
    $x_{i}^{1} = t$ means that $x_{i}=1$;
    \item variables 
    $y_{i}^0$ and $y_{i}^1$ encode the variable $y_{i}$;
    $y_{i}^{0} = t$ means that $y_{i}=0$,
    $y_{i}^{1} = t$ means that $y_{i}=1$.
\end{enumerate}

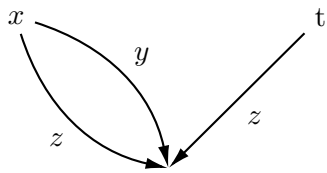
\begin{figure}[H]
\begin{center}
\begin{tikzpicture}[scale=1]

\node at (0,0) (xi) {$x$};
\draw[-{Latex[length=3mm, width=1.5mm]},thick, bend left] (xi) to node[auto]{$y$} (2,-2);
\draw[-{Latex[length=3mm, width=1.5mm]}, thick, bend right] (xi) to node[below left]{$z$} (2,-2);

\node at (4,0) (true) {t};

\draw[-{Latex[length=3mm, width=1.5mm]}, thick] (true) to node[auto]{$z$} (2,-2);

\end{tikzpicture}
\end{center}

\caption{An example of a graph defining constraints.}
\label{GraphExample}
\end{figure}

\begin{figure*}
\begin{center}
\begin{tikzpicture}[scale=1]


\node at (0,0) (false) {f};
\draw[-{Latex[length=3mm, width=1.5mm]},thick, bend left] (false) to node[auto]{$y_{1}^{1}$} (1,-1);
\draw[-{Latex[length=3mm, width=1.5mm]},thick, bend right] (false) to node[below left]{$y_{1}^{0}$} (1,-1);
\draw[-{Latex[length=3mm, width=1.5mm]},thick, bend left] (1,-1) to node[auto]{$y_{2}^{1}$} (2,-2);
\draw[-{Latex[length=3mm, width=1.5mm]}, thick, bend right] (1,-1) to node[below left]{$y_{2}^{0}$} (2,-2);
\draw[-, dashed, thick] (2,-2) to  (3,-3);
\draw[-{Latex[length=3mm, width=1.5mm]},thick, bend left] (3,-3) to node[auto]{$y_{n}^{1}$} (4,-4);
\draw[-{Latex[length=3mm, width=1.5mm]}, thick, bend right] (3,-3) to node[below left]{$y_{n}^{0}$} (4,-4);

\draw[-{Latex[length=3mm, width=1.5mm]}, thick] (4,-4) to node[below left]{$x_{1}^{1}$} (5,-5);
\draw[-{Latex[length=3mm, width=1.5mm]}, thick] (5,-5) to node[below left]{$x_{2}^{1}$} (6,-6);
\draw[-, dashed, thick] (6,-6) to  (7,-7);
\draw[-{Latex[length=3mm, width=1.5mm]}, thick] (7,-7) to node[below left]{$x_{n}^{1}$} (8,-8);

\node at (13,-3) (true) {t};

\draw[-{Latex[length=3mm, width=1.5mm]}, thick] (true) to node[auto]{$x_{1}^{0}$} (12,-4);
\draw[-{Latex[length=3mm, width=1.5mm]}, thick] (12,-4) to node[auto]{$x_{2}^{0}$} (11,-5);
\draw[-, dashed, thick] (11,-5) to  (10,-6);
\draw[-{Latex[length=3mm, width=1.5mm]}, thick] (10,-6) to node[auto]{$x_{n}^{0}$} (9,-7);
\draw[-{Latex[length=3mm, width=1.5mm]}, thick] (9,-7) to node[auto]{$z$} (8,-8);


\end{tikzpicture}
\end{center}
\caption{Set of constraints $\mathcal C_{0}$.}
\label{fig:C0}
\end{figure*}

\begin{figure*}
\begin{center}
\begin{tikzpicture}[scale=1]

\node at (0,0) (xi) {$x_{i}^{0}$};
\draw[-{Latex[length=3mm, width=1.5mm]},thick, bend left] (false) to node[auto]{$y_{i}^{1}$} (1,-1);
\draw[-{Latex[length=3mm, width=1.5mm]},thick, bend right] (false) to node[below left]{$y_{i}^{0}$} (1,-1);
\draw[-{Latex[length=3mm, width=1.5mm]},thick, bend left] (1,-1) to node[auto]{$y_{i+1}^{1}$} (2,-2);
\draw[-{Latex[length=3mm, width=1.5mm]}, thick, bend right] (1,-1) to node[below left]{$y_{i+1}^{0}$} (2,-2);
\draw[-, dashed, thick] (2,-2) to  (3,-3);
\draw[-{Latex[length=3mm, width=1.5mm]},thick, bend left] (3,-3) to node[auto]{$y_{n}^{1}$} (4,-4);
\draw[-{Latex[length=3mm, width=1.5mm]}, thick, bend right] (3,-3) to node[below left]{$y_{n}^{0}$} (4,-4);

\draw[-{Latex[length=3mm, width=1.5mm]}, thick] (4,-4) to node[below left]{$x_{i+1}^{1}$} (5,-5);
\draw[-{Latex[length=3mm, width=1.5mm]}, thick] (5,-5) to node[below left]{$x_{i+2}^{1}$} (6,-6);
\draw[-, dashed, thick] (6,-6) to  (7,-7);
\draw[-{Latex[length=3mm, width=1.5mm]}, thick] (7,-7) to node[below left]{$x_{n}^{1}$} (8,-8);

\node at (13,-3) (true) {t};

\draw[-{Latex[length=3mm, width=1.5mm]}, thick] (true) to node[auto]{$x_{i+1}^{0}$} (12,-4);
\draw[-{Latex[length=3mm, width=1.5mm]}, thick] (12,-4) to node[auto]{$x_{i+2}^{0}$} (11,-5);
\draw[-, dashed, thick] (11,-5) to  (10,-6);
\draw[-{Latex[length=3mm, width=1.5mm]}, thick] (10,-6) to node[auto]{$x_{n}^{0}$} (9,-7);
\draw[-{Latex[length=3mm, width=1.5mm]}, thick] (9,-7) to node[auto]{$z$} (8,-8);


\end{tikzpicture}
\end{center}
\caption{Set of constraints $\mathcal C_{i}$.}
\label{fig:Ci}
\end{figure*}


\begin{figure}[H]
\begin{center}
\begin{tikzpicture}[scale=1]

\node at (0,0) (xi) {$x_{n-1}^{0}$};
\draw[-{Latex[length=3mm, width=1.5mm]},thick, bend left] (xi) to node[auto]{$y_{n-1}^{1}$} (1,-1);
\draw[-{Latex[length=3mm, width=1.5mm]},thick, bend right] (xi) to node[below left]{$y_{n-1}^{0}$} (1,-1);
\draw[-{Latex[length=3mm, width=1.5mm]},thick, bend left] (1,-1) to node[auto]{$y_{n}^{1}$} (2,-2);
\draw[-{Latex[length=3mm, width=1.5mm]}, thick, bend right] (1,-1) to node[below left]{$y_{n}^{0}$} (2,-2);

\draw[-{Latex[length=3mm, width=1.5mm]}, thick] (2,-2) to node[below left]{$x_{n}^{1}$} (3,-3);

\node at (5,-1) (true) {t};

\draw[-{Latex[length=3mm, width=1.5mm]}, thick] (true) to node[auto]{$x_{n}^{0}$} (4,-2);
\draw[-{Latex[length=3mm, width=1.5mm]}, thick] (4,-2) to node[auto]{$z$} (3,-3);


\end{tikzpicture}
\end{center}
\caption{Set of constraints $\mathcal C_{n-1}$.}
\label{fig:Cn1}
\end{figure}

\begin{figure}[H]
\begin{center}
\begin{tikzpicture}[scale=1]

\node at (0,0) (xi) {$x_{n}^{0}$};
\draw[-{Latex[length=3mm, width=1.5mm]},thick, bend left] (xi) to node[auto]{$y_{n}^{1}$} (2,-2);
\draw[-{Latex[length=3mm, width=1.5mm]}, thick, bend right] (xi) to node[below left]{$y_{n}^{0}$} (2,-2);

\node at (4,0) (true) {t};

\draw[-{Latex[length=3mm, width=1.5mm]}, thick] (true) to node[auto]{$z$} (2,-2);

\end{tikzpicture}
\end{center}

\caption{Set of constraints $\mathcal C_{n}$.}
\label{fig:Cn}
\end{figure}

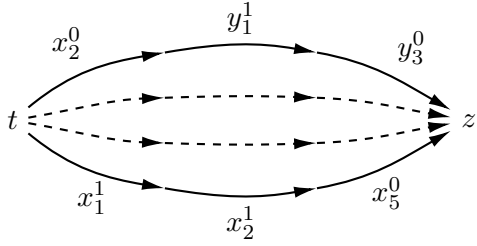
\begin{figure}[H]
\begin{center}
\begin{tikzpicture}[scale=1]

\node at (0,0) (true) {$t$};
\node at (6,0) (z) {$z$};
\draw[-{Latex[length=3mm, width=1.5mm]},thick, out = 40, in = 190] (true) to node[auto]{$x_{2}^{0}$} (2,0.9);
\draw[-{Latex[length=3mm, width=1.5mm]},thick, out = 10, in = 170] (2,0.9) to node[auto]{$y_{1}^{1}$} (4,0.9);
\draw[-{Latex[length=3mm, width=1.5mm]},thick, out = -10, in =150] (4,0.9) to node[auto]{$y_{3}^{0}$} (z);

\draw[-{Latex[length=3mm, width=1.5mm]},thick, out = -40, in = -190] (true) to node[below]{$x_{1}^{1}$} (2,-0.9);
\draw[-{Latex[length=3mm, width=1.5mm]},thick, out = -10, in = -170] (2,-0.9) to node[below]{$x_{2}^{1}$} (4,-0.9);
\draw[-{Latex[length=3mm, width=1.5mm]},thick, out = 10, in =-150] (4,-0.9) to node[below]{$x_{5}^{0}$} (z);

\draw[-{Latex[length=3mm, width=1.5mm]},thick, dashed, out = 10, in = 182] (true) to node[below]{} (2,0.3);
\draw[-{Latex[length=3mm, width=1.5mm]},thick, dashed, out = 2, in = 178] (2,0.3) to node[below]{} (4,0.3);
\draw[-{Latex[length=3mm, width=1.5mm]},thick, dashed, out = -2, in =170] (4,0.3) to node[below]{} (z);

\draw[-{Latex[length=3mm, width=1.5mm]},thick, dashed, out = -10, in = -182] (true) to node[below]{} (2,-0.3);
\draw[-{Latex[length=3mm, width=1.5mm]},thick, dashed, out = -2, in = -178] (2,-0.3) to node[below]{} (4,-0.3);
\draw[-{Latex[length=3mm, width=1.5mm]},thick, dashed, out = 2, in =-170] (4,-0.3) to node[below]{} (z);

\end{tikzpicture}
\end{center}

\caption{Set of constraints $\mathcal F$.}
\label{fig:F}
\end{figure}


\begin{thm}
 Formulas $\Phi$ and $\Psi$ are equivalent.
\end{thm}

\begin{lem}
$\Phi\rightarrow\Psi$.
\end{lem}
\begin{proof}
In the following, we make no assumption that $t\neq f$, though the more interesting cases arise when this is true. Suppose we have a strategy for the existential player (EP) in $\Phi$.
Let us define a strategy for the existential player (EP) in $\Psi$.
First, we want at most one of the two values $x_{i}^{0}$ and $x_{i}^{1}$ to be equal to $t$,
and at most one of $y_{i}^{0}$ and $y_{i}^{1}$ to be equal to $t$.
Assume that the above rule was violated for the first time by the universal player (UP) in $\Psi$
when he plays $y_{i}^{0} = t$ and $y_{i}^{1} = t$, 
or when he plays $x_{i}^{1} = t$ after $x_{i}^{0} = t$.
In both cases the winning strategy for the EP is 
to choose $x_{i+1}^{0} = \dots = x_{n}^{0} = z = t$.
Thus, we may assume that the above rule always holds.

Then we interpret $x_{i}^{0}=t$ as $x_{i}=0$,
$x_{i}^{1}=t$ as $x_{i}=1$,
$y_{i}^{0}=t$ as $y_{i}=0$,
$y_{i}^{1}=t$ as $y_{i}=1$.
It will be convenient for us to assume that 
$x_{i}=0$ if both 
$x_{i}^{0}$ and $x_{i}^{1}$ are different from $t$.
Similarly, we assume that 
$y_{i}=0$ if both 
$y_{i}^{0}$ and $y_{i}^{1}$ are different from $t$.

Then the EP in $\Psi$ should play according to the strategy of the EP in $\Phi$,
that is, if the EP in $\Phi$ chooses $x_{i} =0$, then 
the EP in $\Psi$ should choose $x_{i}^{0}=t$, 
otherwise 
the EP in $\Psi$ should choose $x_{i}^{0}=f$.
Additionally, the EP in $\Psi$ chooses $z=f$, and the remaining values (for the variables that are not labeled in figures) can be chosen uniquely.

It is not hard to see that the constraints in $\mathcal F$ are equivalent to the 3-CNF part of $\Phi$. This is because at least one of the edge labels in the path must not be evaluated to $t$, which corresponds to one of the literals being true.
\end{proof}

\begin{lem}
$\Psi\rightarrow\Phi$.
\end{lem}
\begin{proof}
Suppose we have a winning strategy for the EP in $\Psi$. 
To define a winning strategy for the EP in $\Phi$,
we will show how the UP in $\Psi$ should play 
and how moves of the EP in $\Psi$ should be interpreted 
by the EP in $\Phi$.
First, the UP in $\Psi$ assigns two different values to $t$ and $f$.
If the EP plays $x_{i}^{0}=t$ then the EP in $\Phi$ plays $x_{i}=0$,
if the EP plays $x_{i}^{0}\neq t$ then the EP in $\Phi$ plays $x_{i}=1$.
Let us show how the UP in $\Psi$ should play.
Let $d$ be the last value in the sequence $f,x_{1}^{0},x_{2}^{0},\dots,x_{i}^{0}$ different from $t$.
He plays $x_{i}^{1} = d$ if $x_{i}^{0}=t$, and $x_{i}^{1} = t$ if $x_{i}^{0}\neq t$.
He plays $y_{i}^0=t$ and $y_{i}^1=d$ if the UP in $\Phi$ plays $y_{i}= 0$.
He plays $y_{i}^0=d$ and $y_{i}^1=t$ if the UP in $\Phi$ plays $y_{i}= 1$.

Let us show that if the EP in $\Phi$ follows this strategy then all the constraints are satisfied.
Let $i$ be the maximal number such that $x_{i}=1$, which means that $x_{i}^{0}\neq t$. Put $d = x_{i}^{0}$.
If $x_j = 0$ for every $j$, then we put $i=0$ and $d= f$.
Since $x_{i+1}=\dots=x_{n}=0$, 
we have 
$x_{i+1}^{0}=\dots=x_{n}^{0}=t$ and 
$x_{i+1}^{1}=\dots=x_{n}^{1}=d$. 
Notice that 
for every $i\in\{i,i+1,\dots, n\}$
exactly one of the two variables $y_{i}^{0}$ and $y_{i}^{1}$
is equal to $d$.
Hence, the constraints 
from $\mathcal C_{i}$ imply $z\neq t$.
Then the constraints from $\mathcal F$ imply that the 3-CNF part holds, since, for each path from $z$ to $t$, at least one of the edge labels must not be evaluated as $t$.
For instance, the path $x_{2}^0 - y_{1}^{1} - y_{3}^{0}$ 
in Figure \ref{fig:F} implies that at least one of the variables 
$x_{2}^0, y_{1}^{1}, y_{3}^{0}$ does not equal $t$, which means that 
$x_{2}=1$, $y_{1}=0$, or $y_{3}=1$ and 
the constraint $(x_2\vee \overline y_1\vee y_3)$ holds.
\end{proof}
 
\begin{cor}
QCSP$(\mathbb{N};x=y\rightarrow y=z)$ is PSpace-complete.
\label{cor:main}
\end{cor}

Let us remark that an alternative construction of $\Psi$ has $t$ and $f$ quantified existentially (still outermost) but adds the new constraint $t\neq f$ (note that disequality $x\neq y$ can be defined over $I$ by 
 $\forall z\; I(x,y,z)$).

\section{Bounded alternation}

We begin with the most substantial contribution of this section which is the Co-NP membership result.

\subsection{Bounded alternation QCSP$(\mathbb{N};x=y\rightarrow u=v)$ is in Co-NP}

Recall that $R$ is defined by a $\Sigma_{k}$-formula over $\Gamma$ if 
\[
\begin{array}{ll}
R(x_1,\dots,x_n) = & \exists\dots\exists \forall\dots\forall\exists\dots\exists\forall\dots\forall \\  
& R_{1}(\dots)\wedge \dots\wedge R_{s}(\dots),
\end{array}
\] 
where the formula has $(k-1)$ alternations and 
$R_1,\dots,R_s$ are from $\Gamma$.
In our case $\Gamma := (\mathbb{N};x=y\rightarrow u=v, x=y)$ where we emphasise that atoms $x=y$ can also appear.

First, we define inductively what a $k$-proof means.
Suppose
$R$ is defined by a $\Sigma_1$-formula $\Phi$ over $\Gamma$, 
that is 
$$R(x_1,\dots,x_n)=\exists x_{n+1}\dots\exists 
x_{n+m} \;\; C_1\wedge C_2\wedge \dots \wedge C_{s},$$
where $C_{1},C_2,\dots,C_{s}$ are constraints.

Let $P=(e_1,e_2,\dots,e_r)$ be a sequence of equalities on $\{x_1,\dots,x_{n+m}\}$, where $e_r$ is 
an equality on $\{x_1,\dots,x_n\}$.
We say that $P$ is a \emph{0-proof in $\Phi$ of
the equality $e_r$  from some set of equalities $E$ on $\{x_1,\dots,x_{n}\}$}, 
if for every $i\in[r]$ the equality $e_{i}$ satisfies one of the following conditions:
\begin{enumerate}
    \item $e_{i}$ is from $E$;
    \item $e_{i}$ can be derived by transitivity from 
    $e_{i_1}$ and $e_{i_2}$, where $i_{1},i_{2}<i$;
    \item $\Phi$ contains the constraint $e_{i}$;
    \item $\Phi$ contains the constraint $(e_{j}\rightarrow e_{i})$
    for $j<i$.
\end{enumerate}
Note that we respect the commutativity of equality, in the sense that a derivation of $x=y$ a fortiori gives also $y=x$.

Suppose
$R$ is defined by a $\Sigma_{2k+1}$-formula $\Phi$ over $\Gamma$, 
that is 
\[
\begin{array}{ll}
R(x_1,\dots,x_n)= & \exists x_{n+1}\dots\exists 
x_{n+m} \forall u_1\dots \forall u_t \\  & R'(x_1,\dots,x_{n+m},u_1,\dots,u_t),
\end{array}
\]
where $R'$ is defined by $\Sigma_{2k-1}$-formula $\Phi'$ over $\Gamma$.

A sequence $P=(o_1,\dots,o_r)$ is called \emph{
a k-proof in $\Phi$ of
an equality $e_{r}$ on $\{x_1,\dots,x_{n}\}$ from some set of equalities $E$ on $\{x_1,\dots,x_{n}\}$} if for every $i\in[r]$  the following conditions hold:
\begin{itemize}
    \item $o_i$ is a triple 
    $(e_{i}, E_{i}^{u}, P_{i}^0)$;
    \item $e_{i}$ is an equality on $\{x_1,\dots,x_{n+m}\}$;
    \item $E_{i}^{u}$ is a set of equalities on $\{x_1,\dots,x_{n+m},u_1,\dots,u_t\}$, 
    such that each equality is $u_{i}=z$, where 
    \\ $z\in \{x_1,\dots,x_{n+m},u_1,\dots,u_{i-1}\}$,
    and each variable $u_i$ appears at most once as the left variable of the equality;
    \item $P_i^{0}$ is 
    a $(k-1)$-proof in $\Phi'$ of the equality 
    $e_{i}$ from the set of equalities 
    $E\cup E_{i}^{u}\cup \{e_1,\dots,e_{i-1}\}$. 
\end{itemize}

Similarly, we define a $k$-proof of a contradiction.
A sequence $P=(o_1,\dots,o_r)$ is called \emph{
a k-proof in $\Phi$ of a contradiction from some set of equalities $E$ on $\{x_1,\dots,x_{n}\}$} if for every $i\in[r-1]$  the following conditions hold:
\begin{itemize}
    \item $o_i$ is a triple 
    $(e_{i}, E_{i}^{u}, P_{i}^0)$;
    \item $e_{i}$ is an equality on $\{x_1,\dots,x_{n+m}\}$;
    \item $E_{i}^{u}$ is a set of equalities on $\{x_1,\dots,x_{n+m},u_1,\dots,u_t\}$, 
    such that each equality is $u_{i}=z$, where \\
    $z\in\{x_1,\dots,x_{n+m},u_1,\dots,u_{i-1}\}$, 
    and each variable $u_i$ appears at most once as the left variable of the equality;
    \item $P_i^{0}$ is 
    a $(k-1)$-proof in $\Phi'$ of the equality 
    $e_{i}$ from the set of equalities 
    $E\cup E_{i}^{u}\cup \{e_1,\dots,e_{i-1}\}$. 
\end{itemize}
and $o_{r}$ satisfies one of the following conditions:
\begin{enumerate}
    \item $o_r=(\varnothing, E_{r}^{u}, P_{r}^0)$, 
    where $E_{r}^{u}$ is defined as before, 
    and $P_{r}^{0}$ is a $(k-1)$-proof in $\Phi'$ of a contradiction
    from the set of equalities
    $E\cup E_{r}^{u}\cup \{e_1,\dots,e_{i-1}\}$; 
    \item $o_r=(e_{r}, E_{r}^{u}, P_{r}^0)$, where 
    $e_{r}, E_{r}^{u}, P_{r}^0$ are defined as above but 
    $e_{r}$ is an equality 
    $u_i=z$, where \\ $z\in\{x_1,\dots,x_{n+m},u_1,\dots,u_{i-1}\}$
    and $E_{r}^{u}$ does not contain the equality 
    $u_{i}=z'$ for any $z'$.
\end{enumerate}
 
 Note that there is no such a notion as $0$-proof of a contradiction.
We need to prove the following lemmas about the soundness of our proofs.
\begin{lem}
Suppose 
$R(x_1,\dots,x_n)$ is defined by a $\Sigma_{2k+1}$-formula $\Phi$ over $\Gamma$, 
$P$ is 
a $k$-proof in $\Phi$ of
an equality $e$ on $\{x_1,\dots,x_{n}\}$ from some set of equalities $E$ on $\{x_1,\dots,x_{n}\}$.
Then any tuple from $R$ that satisfies all the equalities from $E$
also satisfies $e$.
\end{lem}
\begin{proof}
We proceed by induction on $k$. At $k=0$, the property holds from the soundness of the steps in the definition of $0$-proof. Suppose it holds at $k=m$, then it holds at $k=m+1$ by the soundness of the steps in the definition of $(m+1)$-proof. Note that each step $(e_i,E^u_i,P^0_i)$ is sound because all of the equalities in $E^u_i$ can be simultaneously realised by the appropriate evaluation of $u_1,\ldots,u_t$.
\end{proof}

\begin{lem}
Suppose 
$R(x_1,\dots,x_n)$ is defined by a $\Sigma_{2k+1}$-formula $\Phi$ over $\Gamma$, 
$P$ is 
a $k$-proof in $\Phi$ of
a contradiction from some set of equalities $E$ on $\{x_1,\dots,x_{n}\}$.
Then $R$ does not contain tuples satisfying all the equalities from $E$.
\end{lem}
\begin{proof}
We proceed by induction on $k$. At $k=1$, the property holds from the soundness of the steps in the definition of $1$-contradiction. In particular, if $e_r$ is of the form $u_i=z$, where $E^u_r$ holds no information on $u_i$, then assigning $u_i$ to be anything other than (the value of) $z$ demonstrates that the formula must be false (a contradiction).

Suppose it holds at $k=m$, then it holds at $k=m+1$ by the soundness of the steps in the definition of $(m+1)$-contradiction (as described in the previous paragraph).
\end{proof}

We now need a technical lemma showing that we have transitivity not only in 
$0$-proofs but also in $k$-proofs.
\begin{lem}\label{TransitivityLemma}
Suppose 
$R(x_1,\dots,x_n)$ is defined by a $\Sigma_{2k+1}$-formula $\Phi$ over $\Gamma$, 
then there exists a proof of the equality 
$x_{i_1} = x_{i_3}$ from the set of equalities 
$\{x_{i_1} = x_{i_2},x_{i_2} = x_{i_3}\}$.
\end{lem}
\begin{proof}
We prove by induction on $k$.
For $k=0$ it follows from the definition. 
for larger $k$ take a one step proof 
$(x_{i_1}=x_{i_3},\varnothing, P_{1}^{0})$, where 
$P_{1}^{0}$ is a $(k-1)$-proof of transitivity, 
which exists by the inductive assumption.
\end{proof}

We now give our principal completeness theorem.
\begin{thm}\label{ExistenceOfAProof}
Suppose $R(x_1,\dots,x_n)$ is defined by a $\Sigma_{2k+1}$-formula $\Phi$ over $\Gamma:=(\mathbb{N};x=y\rightarrow u=v, x=y)$, 
$R(a_1,\dots,a_n)$ does not hold, 
and $E$ is the set of all equalities on 
$\{x_1,\dots,x_n\}$ satisfied by the tuple 
$(a_1,\dots,a_n)$.
Then one of the following conditions holds:
\begin{itemize}
    \item there exists 
a $k$-proof in $\Phi$
of some equality $x_{i}=x_{j}$ from the set of equalities $E$ that is not satisfied by 
$(a_1,\dots,a_n)$, i.e. $a_{i}\neq a_{j}$;
\item there exists 
a $k$-proof in $\Phi$
of a contradiction from the set of equalities $E$.
\end{itemize}
\end{thm}

\begin{proof}
We proceed by induction on $k$.

\textbf{Base}. Suppose $k=0$, then 
$$R(x_1,\dots,x_n)=\exists x_{n+1}\dots\exists 
x_{n+m} \;\; C_1\wedge C_2\wedge \dots \wedge C_{s},$$
where $C_{1},C_2,\dots,C_{s}$ are constraints.
We put in the sequence $P$ all the equalities we can derive. 
If we can derive an equality on $\{x_1,\dots,x_{n}\}$ that is not from $E$, 
then this equality is not satisfied by 
$(a_1,\dots,a_n)$, which completes the proof. 
Assume that we did not derive such an equality.
We need to prove that 
$R$ holds on $(a_1,\dots,a_n)$ to get a contradiction. 
Since we can use transitivity, 
all the variables in $\{x_1,\dots,x_{n+m}\}$ are divided into equivalence classes such that each equality in $P$ is inside one class. 
Note that an equivalence class cannot contain two variables 
$x_{i}$ and $x_{j}$ such that $a_{i}\neq a_{j}$.
Then we send the equivalence class to the value $a_{i}$ if it contains 
a variable $x_{i}$ for $i\in[n]$, 
and we choose a new value for the equivalence class if 
all the variables in it are from $\{x_{n+1},\dots,x_{n+m}\}$.
We claim that this assignment satisfies all the constraints of $\Phi$. 
All the constraints $x_{i}=x_{j}$ are obviously satisfied.
A constraint $x=y\rightarrow u=v$ is satisfied because either 
$x$ and $y$ are from different equivalence classes, or the equality 
$u=v$ should have been added to $P$, and $u,v$ are from one equivalence class.

\textbf{Inductive step}.
Suppose
$R$ is defined by
\[
\begin{array}{ll}
R(x_1,\dots,x_n)= & \exists x_{n+1}\dots\exists 
x_{n+m} \forall u_1\dots \forall u_t \\ & R'(x_1,\dots,x_{n+m},u_1,\dots,u_t),
\end{array}
\]
where $R'$ is defined by $\Sigma_{2k-1}$-formula $\Phi'$ over $\Gamma$.
Again, we add to the sequence $P$ 
of all the equalities on 
$\{x_{1},\dots,x_{n+m}\}$ we can derive using $(k-1)$-proofs.
By Lemma \ref{TransitivityLemma}, we still have transitivity, 
which means that all the variables 
in $\{x_{1},\dots,x_{n+m}\}$ are divided into equivalence classes. 
Again if an equivalence class contains two variables 
$x_{i},x_{j}$ such that $a_{i}\neq a_{j}$, then the first condition of the theorem holds, which finishes the proof. 
Then we send all the variables in the equivalence class to $a_{i}$ if 
it contains $x_{i}$ for $i\in[n]$, 
otherwise we send
all the variables of this class to a new value that never appeared. 
Thus we get an assignment $(x_1,\dots,x_{n+m}) = (b_{1},\dots,b_{n+m}).$
We claim that this assignment satisfies 
\[\forall u_1\dots \forall u_t \;\;\; R'(x_1,\dots,x_{n+m},u_1,\dots,u_t).\]
Assume that it does not, then for some evaluations of 
$(u_1,\dots,u_{t})=(c_1,\dots,c_{t})$ the tuple 
$(b_1,\dots,b_{n+m},c_1,\dots,c_{t})$ is not in  
the relation $R'$.
By $E^{u}$ we denote the set of equalities 
of the form $u_{i} = z$ that are chosen as follows. 
If $c_{i}$ appears in the tuple 
$(b_1,\dots,b_{n+m},c_1,\dots,c_{i-1})$ then we add the equality 
$u_{i} = z$, where $z$ is the variable corresponding to 
the first appearance of $c_{i}$ in the tuple.
By the inductive assumption applied to $R'$ we have one of the following cases (here we cheated a bit because we didn't add all the equalities 
satisfied by the tuple $(b_1,\dots,b_{n+m},c_1,\dots,c_{i-1})$ to $E^{u}$ but they can be derived by transitivity). 

Case 1. There is  
a $(k-1)$-proof of a contradiction, which gives us 
a $k$-proof of a contradiction.

Case 2. There is 
a $(k-1)$-proof of some equality $z=z'$ that is violated by 
our assignment. 
Since we have transitivity, we may choose variables $z$ and $z'$ so that 
variables that appear earlier in the list 
$x_1,\dots,x_{n+m},u_1,\dots,u_t$ do not have the same assignment as $z$ and
$z'$ respectively. 
Consider two cases. 

Case 2A. The equality $z=z'$ is on the set 
$\{x_1,\dots,x_{n+m}\}$. 
This means that this equality connects two different equivalence classes of variables and should have been added to $P$. Contradiction.

Case 2B. The equality is of the form
$u_i = z'$, where 
$z'\in \{x_1,\dots,x_{n+m},u_1,$ $\dots,u_{i-1}\}$ (we can always switch $z$ and $z'$). This gives us a $k$-proof of a contradiction (see item 2 of the definition).
\end{proof}

It remains to evaluate the size of proof.
We assume that we never repeat the equalities $e_{i}$ because it doesn't make any sense. So, we forbid such proofs.


\begin{lem}\label{SizeOfProof}
Suppose $R$ is defined by a $\Sigma_{2k+1}$-formula $\Phi$ over $\Gamma:=(\mathbb{N};x=y\rightarrow u=v, x=y)$, 
$\Phi$ contains $\ell$ different variables, and $P$ is a $k$-proof of
an equality or a contradiction, then 
the length of $P$ is smaller than $10^{k+1}\cdot \ell^{2k+3}$.
\end{lem}

\begin{proof}
We denote the upper bound on the length of a $k$-proof 
for a formula with $\ell$ variables by 
$L_{\ell}(k)$.
We prove by induction on $k$. 
For $k=0$ we just list equalities, we can have at most $\ell^2$ equalities, 
thus to write the proof we need 
$\ell^2\cdot (2\log_{2}\ell+3)\le 10 \ell^{3}$ symbols.

For larger $k$ we again have at most $\ell^2$ equalities, 
but here we also define $E_{i}^{u}$ and a $(k-1)$-proof for each $i$. 
Thus we have 
\begin{align*}
L_{\ell}(k) & \le \ell^2 ((2\log_{2}\ell+3)
+ \ell^2\cdot (2\log_{2}\ell+3) +
L_{\ell}(k-1))\\
& \leq \ell^2 ((2\log_{2}\ell+3)
+ \ell^2\cdot (2\log_{2}\ell+3) +
10^{k}\ell^{2k+1}) \\
& \le
10^{k+1}\ell^{2k+3}
\end{align*}
\end{proof}


\begin{thm}
$\Sigma_{2k+1}$-$\QCSP(\mathbb{N};x=y\rightarrow u=v)$
is in Co-NP.
\label{thm:zhuk-co-NP}
\end{thm}
\begin{proof}
The sentence defines a null-ary relation $R$, which is either true or false. Since the sentence doesn't have free variables, if it is false then 
by Theorem \ref{ExistenceOfAProof}
there should be a $k$-proof of a contradiction. 
By Lemma \ref{SizeOfProof}, the size of this proof is of polynomial size in the size of the sentence. 
It is clear that the proof can be verified in polynomial time. 
Hence, whenever the sentence doesn't hold we have a polynomial-size proof of a contradiction which puts the problem in Co-NP.
\end{proof}

\subsection{Transformation of a quantified conjunctive formula to an equivalent $\Pi_2$-formula.}

Let us commence with a trivial observation about equality languages.
\begin{lem}
If $\Gamma$ is an equality language and $\phi$ is an input for $\QCSP(\Gamma)$ on $n$ variables, then $\phi$ is a yes-instance iff it is a yes-instance when all variable quantification is relativised to the finite subset $[n]$.
\label{lem:trivial-equality}
\end{lem}
\begin{proof}
The simplest way to see this is through the game between Universal and Existential. As the game progresses through the quantifier order, we update some bijection on the domain of $\Gamma$, to assume, without loss of generality, that all variables are evaluated within $[n]$. The bijection begins as the identity and is updated whenever an element is played outside of $[n]$ by swapping that element with an unplayed element within $[n]$.
\end{proof}

It is well-known on finite domains that any property that is definable in quantified conjunctive logic is already definable by a $\Pi_2$ formula of quantified conjunctive logic (this was first noted explicitly in \cite{HubieSIGACT}). Zhuk gave another proof of this in \cite{abs-2110-09504}, which we reproduce here mildly modified so as to work also for the case equality languages.

Suppose 
$\Psi = \exists y_1\forall x_1\exists y_2\forall x_2 \dots
\exists y_n\forall x_{n} \Phi$.
We define a transformation 
$\zeta$ that transforms $\Psi$ 
into an equivalent $\Pi_{2}$-formula. 
Note that 
$\zeta(\Psi)$ is of exponential size in the size of $\Psi$. 

For each tuple 
$(a_1,\dots,a_n)\in [2n]^{n}$ 
the formula
$\Phi_{a_1,\dots,a_n}$ is obtained from $\Phi$
by  
\begin{enumerate}
    \item replacement of each variable 
    $y_{i}$ by $y_{i}^{a_1,\dots,a_{i-1}}$ (at $i=1$ keep variable as $y_i$)
    \item replacement of each variable 
    $x_{i}$ by $x_{i}^{a_1,\dots,a_{i}}$
\end{enumerate}

By $\zeta(\Psi)$ we denote the sentence 
such that 
\begin{enumerate}
    \item its quantifier-free part is
$\bigwedge\limits_{(a_{1},\dots,a_{n})\in [2n]^{n}}
\Phi_{a_1,\dots,a_n}$
\item first we universally quantify the $x$-variables, 
then we existentially quantify the $y$-variables.
\end{enumerate}

Let us show that 
$\Psi$ and $\zeta(\Psi)$ are equivalent.

\begin{lem}\label{PsiImpliesZetaPsi}
$\Psi \rightarrow \zeta(\Psi)$.
\end{lem}
\begin{proof}
Suppose 
$(f_{1}(), x_{1},
f_{2}(x_1), x_{2},
f_{3}(x_1,x_2), x_{3},...,)$ is a solution of $\Psi$.
Let us show that a solution of $\zeta(\Psi)$ can be defined by 
$$y_{i}^{a_{1},\dots,a_{i-1}} = 
f_{i}(x_{1}^{a_1},x_{2}^{a_1,a_2},x_{3}^{a_1,a_2,a_{3}},
\dots,x_{i-1}^{a_{1},\dots,a_{i-1}}).$$
In fact, each 
conjunctive formula $\Phi_{a_1,\dots,a_n}$
holds as it is just $\Phi$ with all the variables renamed.
\end{proof}

\begin{lem}\label{ZetaPsiImpliesPsi}
$\zeta(\Psi)\rightarrow \Psi $.
\end{lem}
\begin{proof}
Owing to Lemma~\ref{lem:trivial-equality}, we may assume that all variable quantification of $\Psi$ is relativised to the set $[2n]$.
Set 
$x_{i}^{a_{1},\dots, a_{i}} = 
a_{i}$ for every $i$ and 
$a_{1},\dots, a_{i}\in [2n]$.
Since 
$\zeta(\Psi)$ holds, 
for each assignment of 
the $x$-variables there exists 
a proper assignment of the $y$-variables. 
We define 
a solution to $\Psi$ by 
$(f_{1}(), x_{1},
f_{2}(x_1), x_{2},
f_{3}(x_1,x_2), x_{3},...,)$, 
where 
$f_{i}(a_{1},\dots, a_{i-1}) = 
y_{i}^{a_{1},\dots, a_{i-1}}$.
To check that the conjunctive formula $\Phi$ holds
for $(x_1,\dots,x_{n})=(a_1,\dots,a_n)$ 
we consider $\Phi_{a_1,\dots,a_{n}}$.
\end{proof}
Thus, we have proved the following.
\begin{cor}
Let $\Gamma$ be an equality language. Any property that is definable on $\Gamma$ in quantified conjunctive logic is already definable by a $\Pi_2$ formula of quantified conjunctive logic.
\label{cor:pi2}
\end{cor}

\subsection{The remaining cases}

Let us examine the proof of Theorem 9.1 in \cite{BodirskyC10} which proves that the $\Sigma_3$-QCSP$(\mathbb{N};x=y\rightarrow y=z)$ is Co-NP-hard. The reduction is from the complement of \emph{monotone 3-SAT} and the constructed instances of $\Sigma_3$-QCSP$(\mathbb{N};x=y\rightarrow y=z)$ have quantifier prefix $\exists^2\forall^*\exists^*$. Indeed, the first two existential variables are $b_0$ and $b_1$ and are constrained to be distinct. Suppose they are instead universally quantified, then the proof still works as we shall now see.
\begin{thm}[\cite{BodirskyC10}]
$\Pi_2$-QCSP$(\mathbb{N};x=y\rightarrow y=z)$ is Co-NP-hard.
\label{thm:co-np-hard}
\end{thm}
\begin{proof}
We reproduce the construction from \cite{BodirskyC10} with a small amendment as to how we treat the special variables $b_0$ and $b_1$. We reduce from the complement of the problem \emph{Monotone 3-Satisfiability} (MON-3-SAT). In this version of $3$-satisfiability all clauses are either purely negative or purely positive. Let $\phi$ be an instance of MON-3-SAT with variables $v_1,\ldots,v_n$, $l$ negative clauses and $m$ positive clauses. We build an instance $\psi$ of $\Pi_2$-QCSP$(\mathbb{N};x=y\rightarrow y=z)$ with variables $v_1,\ldots,v_n$, $N_1,\ldots,N_l,N'_2,\ldots,N'_l$ and $P_1,\ldots,P_m,P'_2,\ldots,P'_m$, and two additional variables $b_0$ and $b_1$. The instance has quantifier prefix
\[ \forall b_0 \forall b_1 \forall v_1,\ldots,v_n \exists N_1,\ldots,N_l,N'_2,\ldots,N'_l,P_1,\ldots,P_m,P'_2,\ldots,N'_m.\]
The quantifier-free part of $\psi$ is the conjunction of the following:
\begin{enumerate}
    \item $(v_i=b_0 \rightarrow b_0=N_h) \wedge (v_j=b_0 \rightarrow b_0=N_h) \wedge (v_k=b_0 \rightarrow b_0=N_h)$ for each negative clause $N_h=(\neg v_i \vee \neg v_j \vee \neg v_k)$.
    \item $(v_i=b_1 \rightarrow b_1=P_h) \wedge (v_j=b_1 \rightarrow b_1=P_h) \wedge (v_k=b_1 \rightarrow b_1=P_h)$ for each positive clause $P_h=(v_i \vee v_j \vee v_k)$.
    \item $(N_1=N_2\rightarrow N_2=N'_2) \wedge \bigwedge_{h=3,\ldots,l}(N'_{h-1}=N_h\rightarrow N_h=N'_h)$.
    \item $(P_1=P_2\rightarrow P_2=P'_2) \wedge \bigwedge_{h=3,\ldots,m}(P'_{h-1}=P_h\rightarrow P_h=P'_h)$.
    \item $N'_l=P'_m$.
\end{enumerate}
(We assume that the formula $\phi$ contains at least two positive clauses and two negative clauses.) 

($\phi$ is a yes-instance of MON-3-SAT implies $\psi$ is a no-instance of our QCSP.)
Let us assume that $b_0$ and $b_1$ are evaluated to different elements. Assume that the original formula is satisfiable with some assignment $f:\{v_1,\ldots,v_n\}\rightarrow \{0,1\}$. Consider the induced assignment to the remaining universally quantified variables when $v_i$ is set to $b_{f(v_i)}$. Since $f$ satisfies all of the clauses, it must holds that $N_h=b_0$ for all $h \in [l]$ and $P_h=b_1$ for all $h \in [m]$. It must further hold that $N'_h=b_0$ for all $h\in \{2,\ldots,l\}$ and that $P'_h=b_1$ for all $h\in \{2,\ldots,m\}$. But this implies that $N'_l\neq P'_m$ and we have that the formula $\psi$ is false.

($\phi$ is a no-instance of MON-3-SAT implies $\psi$ is a yes-instance of our QCSP.) Now consider that $\phi$ is false. If $b_0$ and $b_1$ are evaluated to the same element, then $\psi$ will be true. Let us now assume that they are evaluated to different elements. Consider some assignment $g:\{v_1,\ldots,v_n\}\rightarrow \{0,1\}$. Let us say that a negative clause $N_h$ is not satisfied by $g$ if none of its variables is set by $g$ to be $b_0$; similarly, let us say that a positive clause $P_h$ is not satisfied by $g$ is none of its variables is set by $g$ to be $b_1$. By assumption there exists a negative clause not satisfied by $g$ or a positive clause not satisfied by $g$. Let us discuss the negative case as the positive case is dual. If it is a negative clause $N_h$, then all constraints in (1) above can be satisfied without setting $N_h$ to $b_0$. It follows that the constraints in (3) can be satisfied with $N'_l$ set to $b_1$. Now setting all of the variables $P_h$ and $P'_h$ to $b_1$ results in all the constraints being satisfied. It follows that the formula $\psi$ is true.
\end{proof}
\begin{cor}
Suppose that $x=y\rightarrow y=z$ is definable in an equality language $\Gamma$ in a $\Pi_2$ formula of quantified conjunctive logic. Then, $\Pi_2$-QCSP$(\Gamma)$ is Co-NP-hard.
\label{cor:co-np-hard-2}
\end{cor}
\begin{proof}
Let us follow the proof of Theorem~\ref{thm:co-np-hard}. It is clear to see that we can plug in $\Pi_2$ definitions of $x=y\rightarrow y=z$ innermost and then commute the universal quantifiers across the conjunction to obtain an instance that is at worst $\Pi_4$. Let us now argue that we can draw all the existential quantifiers, over the $N_h,N'_h,P_h,P'_h$, innermost to result in an instance that is $\Pi_2$. We need to check that in the case that $\phi$ is a yes-instance of MON-3-SAT, that this still implies $\psi$ is a no-instance of our QCSP, even when the existential quantifiers are further in (which could advantage Existential). But indeed this is true because they have been forced by the outermost universal quantifiers.
\end{proof}

A cursory examination of the proof of Corollary~\ref{cor:main}, where one starts with a universal sentence of \emph{Quantified-3-Satisfiability}, shows that it produces a $\Pi_2$ instance of QCSP$(\mathbb{N};x=y\rightarrow y=z)$. However, the universal fragment of Quantified-3-Satisfiability is in P (since the quantifier free part is in 3-CNF).

We now develop some lemmas before addressing the case of maximal complexity, that is, rising in the polynomial hierarchy.
\begin{lem}
$(P_1(x_1,x'_1) \vee \ldots \vee P_m(x_m,x'_m))$, where each $P_i$ asserts either equality or disequality, is pp-definable in $(\mathbb{N};\neq,x=y \vee u=v)$
\label{lem:build-disjunction}
\end{lem}
\begin{proof}
$\exists v' (x=y \vee u=v') \wedge v\neq v'$ pp-defines $(x=y \vee u \neq v)$. Then the desired pp-definition goes as follows.
\[
\begin{array}{cl}
\exists y_1,\ldots,y_{m+1} & (P_1(x_1,x'_1) \vee y_1=y_2)) \wedge \\
& (P_2(x_2,x'_2) \vee y_2=y_3)) \ \wedge \\
& \vdots \\
& (P_m(x_m,x'_m) \vee y_m=y_{m+1})) \ \wedge \\
& y_1 \neq y_{m+1} \\
\end{array}
\]
\end{proof}

\begin{lem}
The $\Pi_k$-$\QCSP(\mathbb{N};\neq,x=y \vee u=v)$ is $\Pi^{\mathrm{P}}_k$-hard.
\label{lem:hardest}
\end{lem}
\begin{proof}
We will give a reduction from $\Pi_k$-\emph{Quantified-Not-All-Equal-3-SAT} ($\Pi_k$-QNAE-3-SAT).
 Let $I$ be an instance of $\Pi_k$-QNAE-3-SAT. Each variable $v$ in $I$ becomes a pair of variables $v,v'$ in $\Phi$, an instance of $\Pi_k$-$\QCSP(\mathbb{N};\neq,x=y \vee u=v)$, in which $v=v'$ encodes that the variable is true and $v\neq v'$ encodes that the variable is false.
 
Now we can define the not-all-equal-3-sat predicate in disjunctive normal form (DNF) as:
\[
\begin{array}{l}
(x=x' \wedge y=y' \wedge z\neq z') \ \vee \\
(x=x' \wedge y \neq y' \wedge z = z') \ \vee \\
(x=x' \wedge y\neq y' \wedge z\neq z') \ \vee \\
(x \neq x' \wedge y=y' \wedge z=z') \ \vee \\
(x \neq x' \wedge y=y' \wedge z\neq z') \ \vee \\
(x \neq x' \wedge y\neq y' \wedge z=z') \\
\end{array}
\]
This expands in CNF to a set of clauses each with six literals that would be susceptible to Lemma~\ref{lem:build-disjunction}. It is this conjunction that defines our clause gadget. It is the conjunction of all the clause gadgets that define the quantifier-free part of $\Phi$. The quantification from $I$ translates simply, $\forall v$ becomes $\forall v, v'$ and $\exists v$ becomes $\exists v, v'$.

($I$ in $\Pi_k$-QNAE-3-SAT implies $\Phi$ in $\Pi_k$-$\QCSP(\mathbb{N};\neq,x=y \vee u=v)$.) The winning strategy for Existential on $\Phi$ is to choose some $x,x'$ identical in $\mathbb{N}$, if $x$ had been chosen true in $I$; and some $x,x'$ distinct in $\mathbb{N}$, if $x$ had been chosen false in $I$. The universal variables similarly translate back from $\Phi$ into $I$.

($\Phi$ in $\Pi_k$-$\QCSP(\mathbb{N};\neq,x=y \vee u=v)$ implies $I$ in $\Pi_k$-QNAE-3-SAT.) The winning strategy for Existential on $I$ is to choose $x$ true, if $x,x'$ had been chosen identical in $I$; and some $x$ false, if $x,x'$ had been chosen identical in $I$. The universal variables similarly translate back from $I$ into $\Phi$. 
\end{proof}

\begin{lem}
$\CSP(\mathbb{N};0,1,x=y \vee y=z)$ is NP-hard.
\label{lem:NP-hard-bis}
\end{lem}
\begin{proof}
Consider the 2-element (Boolean) language $\Delta$ on domain $\{0,1\}$ with relations $\neq$ and $x=y\vee y=z$. Since $\neq$ is not preserved by a semilattice ($\wedge$ or $\vee$) or constant operation, and $x=y\vee y=z$ is not preserved by majority or minority, it follows from Schaefer Theorem (see the modern presentation of Theorem~3.20 in \cite{HubieSIGACT})  that $\CSP(\Delta)$ is NP-hard. We will reduce from $\CSP(\Delta)$ to $\CSP(\mathbb{N};0,1,x=y \vee y=z)$ to prove our result. We can define in $(\mathbb{N};0,1,x=y \vee y=z)$ the subdomain $\{0,1\}$ with $y=0 \vee y=1$. We can define on $\{0,1\}$ the binary disequality relation by expanding $(x=0 \wedge y=1)\vee(x=1 \wedge y=0)$ to
\[ (x=0 \vee x=1) \wedge (x=0 \vee y=0) \wedge  (y=1 \vee x=1) \wedge (y=1 \vee y=0).\]
Finally, we can define $x=y \vee y=z$ by itself. Thus, we reduce an instance $\phi$ of $\CSP(\Delta)$ to an instance $\psi$ of $\CSP(\mathbb{N};0,1,x=y \vee y=z)$ in the following fashion. The set of variables stays the same and we replace each relation in $\phi$ by its definition in $(\mathbb{N};0,1,x=y \vee y=z)$. Finally, for each variable $v$, we add the constraint $v=0 \vee v=1$. That $\phi$ is a yes-instance iff $\psi$ is a yes-instance follows immediately by construction.
\end{proof}

\begin{cor}
$\Pi_2$-$\QCSP(\mathbb{N};x=y \vee y=z)$ is NP-hard.
\label{cor:NP-hard-bis}
\end{cor}
\begin{proof}
We reduce from $\CSP(\mathbb{N};0,1,x=y \vee y=z)$ by adding two new variables $b_0$ and $b_1$ universally quantified outermost and then substituting all instances of $0$ with $b_0$ and all instances of $1$ with $b_1$. Note that, when $b_0$ and $b_1$ are evaluated as equal, the remainder of the sentence will always be true. 
\end{proof}
\begin{lem}
Let $\Gamma$ be an equality language that is positive and not negative. There exists a finite set of constants $[m]$ so that there is a pp-definition in $(\Gamma;1,\ldots,m)$ of $x=y \vee y=z$ over $\mathbb{N}\setminus [m]$.
\label{lem:NP-hard-bis-general}
\end{lem}
\begin{proof}
By Lemma \ref{lem:small} and Corollary \ref{cor:pi2}, there is a $\Pi_2$ quantified conjunctive formula $\theta$ that defines $x=y \vee y=z$ over $\Gamma$. Suppose it has $m$ universal variables. Obtain $\theta'$ from $\theta$ by substituting, in turn, the $m$ universal variables with the constants $1,\ldots,m$. Since $\Gamma$ is positive, $\theta'$ defines $x=y \vee y=z$ on the set $\mathbb{N}\setminus [m]$ (Universal's optimal strategy on $\theta$ will always involve playing new elements).
\end{proof}
\begin{cor}
Let $\Gamma$ be a equality language that is positive and not negative. $\Pi_2$-$\QCSP(\Gamma)$ is NP-hard.
\label{cor:NP-hard-bis-general}
\end{cor}
\begin{proof}
We use the reduction from Lemma~\ref{lem:NP-hard-bis} rather as we did in Corollary~\ref{cor:NP-hard-bis} but now we have $m+2$ outermost universally quantified variables. Two of them will play the role of $0$ and $1$ from Lemma~\ref{lem:NP-hard-bis} and then $m$ further of them will play the role of the constants $[m]$ in Lemma~\ref{lem:NP-hard-bis-general}. Note that because $\Gamma$ is a positive language, an optimal strategy for Universal will always be to play these variables on distinct elements of $\mathbb{N}$, whereupon the correctness of the simulation of Lemma~\ref{lem:NP-hard-bis} is guaranteed by Lemma~\ref{lem:NP-hard-bis-general}. 
\end{proof}

\subsection{The classification}


\begin{lem}
If $\Gamma$ is an equality language containing some relation that does not have a Horn definition, then the relation $x=y\vee u=v$ is definable in $\Gamma$ by a quantified conjunctive formula.
\label{lem:hauptlemma}
\end{lem}
\begin{proof}
 If $\Gamma$ is preserved by an essential operation with infinite image $f$, then by Proposition 37 in \cite{BodirskyCP10}
 $f$ preserves $\neq$ or $\Gamma$ is preserved by all operations. In the latter case every relation in $\Gamma$ may be defined by a conjunction of equalities. Hence it is in particular Horn. By Proposition 43 in \cite{BodirskyCP10}, $\Gamma$ is Horn also in the case where $f$ preserves $\neq$. It contradicts the assumption and implies that $\Gamma$ is preserved by essentially unary operations only.
In particular, the surjective polymorphisms of $\Gamma$ preserve the relation $x=y \vee u=v$. But then Lemma 4.1 in  \cite{BodirskyC10} gives a quantified conjunctive definition.
\end{proof}

\begin{lem}
If $\Gamma$ is an equality language that is positive but not negative, then the relation $x=y\vee u=v$ is definable in $\Gamma$ by a quantified conjunctive formula.
\label{lem:small}
\end{lem} 

\begin{proof}
We follow a line of reasoning similar to that in the proof of Theorem 7.1 in \cite{BodirskyC10}, except that the relation to be defined is $x =y \vee u = v$
not $x = y \vee y = z$. If Pol$(\Gamma)$ contains an essential operation $f$ with infinite image, then as in the proof of Lemma~\ref{lem:hauptlemma}, we show that $\Gamma$ is Horn, and hence negative. It contradicts the assumption. Thus, all surjective polymorphisms of $\Gamma$ are essentially unary.
Since $x=y \vee u=v$ is preserved by all essentially unary operation, the lemma follows by Lemma 4.1 in  \cite{BodirskyC10}.
\end{proof}
\begin{lem}
If $\Gamma$ is an equality language that is not positive, then the relation $\neq$ is quantified conjunctive definable in $\Gamma$.
\label{lem:tiny}
\end{lem}
\begin{proof}
If sPol$(\Gamma)$ has an operation that violates $\neq$, then Theorems 7.3 and 8.2 from \cite{BodirskyC10} imply that $\Gamma$ is positive. The contrapositive of this yields the result according to \cite{BodirskyC10}.
\end{proof}
The following is Proposition 43 in \cite{BodirskyCP10}.
\begin{lem}
All Horn relations are pp-definable in $x=y \rightarrow u=v$ and $\neq$. 
\label{lem:Horn}
\end{lem}
We are now in a position to prove Theorem~\ref{thm:bounded-alternation}.
\begin{proof}[Proof of Theorem~\ref{thm:bounded-alternation}]
The Logspace cases come a fortiori from \cite{BodirskyC10}, since QCSP$(\Gamma)$ is in Logspace in this case.


If $\Gamma$ is positive but not negative, then $\Pi_k$-QCSP$(\Gamma)$ is in NP by \cite{BodirskyC10}. Moreover, $\Pi_2$-QCSP$(\Gamma)$ is NP-hard by Corollary~\ref{cor:NP-hard-bis-general}.

If $\Gamma$ is Horn, but not negative, then $\Pi_k$-QCSP$(\Gamma)$ is in Co-NP by Theorem~\ref{thm:zhuk-co-NP}, together with Lemma~\ref{lem:Horn}. Moreover, $\Pi_2$-QCSP$(\Gamma)$ is Co-NP-hard by Corollary~\ref{cor:co-np-hard-2}. 

Finally, if $\Gamma$ is not Horn and not positive, then, by Lemmas~\ref{lem:hauptlemma} and \ref{lem:tiny}, $\neq$ and $x=y\vee y=z$ are quantified conjunctive definable, and we know from Corollary~\ref{cor:pi2} even by a $\Pi_2$ formula. It follows from Lemma~\ref{lem:hardest} that $\Pi_{k}$-QCSP$(\Gamma)$ is $\Pi_{k-2}^{\mathrm{P}}$-hard (since the innermost conjunction of $\Pi_2$ formulas may itself be turned into a single $\Pi_2$ formula by commuting conjunction and universal quantification).
\end{proof}

\section{Conclusions}

In the case where $\Gamma$ is not Horn and not positive, we proved that $\Pi_{k}$-QCSP$(\Gamma)$ is $\Pi_{k-2}^{\mathrm{P}}$-hard. It is likely that this can be improved to $\Pi_{k}^{\mathrm{P}}$-hard. However, neither Lemmas~\ref{lem:hauptlemma} and \ref{lem:tiny} can be improved from quantified conjunctive definability to pp-definability. Respective counterexamples are $x=y \rightarrow y=z$ and $x=y \vee y=z \vee x=z$.


\begin{thebibliography}{10}

\bibitem{BartoKOPP17}
{\sc Barto, L., Kompatscher, M., Ols{\'{a}}k, M., Pham, T.~V., and Pinsker, M.}
\newblock The equivalence of two dichotomy conjectures for infinite domain
  constraint satisfaction problems.
\newblock In {\em 32nd Annual {ACM/IEEE} Symposium on Logic in Computer
  Science, {LICS} 2017, Reykjavik, Iceland, June 20-23, 2017\/} (2017), {IEEE}
  Computer Society, pp.~1--12.

\bibitem{BartoP20}
{\sc Barto, L., and Pinsker, M.}
\newblock Topology is irrelevant (in a dichotomy conjecture for infinite domain
  constraint satisfaction problems).
\newblock {\em {SIAM} J. Comput. 49}, 2 (2020), 365--393.

\bibitem{BodirskyC07}
{\sc Bodirsky, M., and Chen, H.}
\newblock Quantified equality constraints.
\newblock In {\em 22nd {IEEE} Symposium on Logic in Computer Science {(LICS}
  2007), 10-12 July 2007, Wroclaw, Poland, Proceedings\/} (2007), {IEEE}
  Computer Society, pp.~203--212.

\bibitem{BodirskyC10}
{\sc Bodirsky, M., and Chen, H.}
\newblock Quantified equality constraints.
\newblock {\em {SIAM} J. Comput. 39}, 8 (2010), 3682--3699.

\bibitem{BodirskyCP10}
{\sc Bodirsky, M., Chen, H., and Pinsker, M.}
\newblock The reducts of equality up to primitive positive interdefinability.
\newblock {\em J. Symb. Log. 75}, 4 (2010), 1249--1292.

\bibitem{BodirskyK08}
{\sc Bodirsky, M., and K{\'{a}}ra, J.}
\newblock The complexity of equality constraint languages.
\newblock {\em Theory Comput. Syst. 43}, 2 (2008), 136--158.

\bibitem{BodirskyK10}
{\sc Bodirsky, M., and K{\'{a}}ra, J.}
\newblock The complexity of temporal constraint satisfaction problems.
\newblock {\em J. {ACM} 57}, 2 (2010), 9:1--9:41.

\bibitem{BodirskyMM18}
{\sc Bodirsky, M., Martin, B., and Mottet, A.}
\newblock Discrete temporal constraint satisfaction problems.
\newblock {\em J. {ACM} 65}, 2 (2018), 9:1--9:41.

\bibitem{BodirskyM18}
{\sc Bodirsky, M., and Mottet, A.}
\newblock A dichotomy for first-order reducts of unary structures.
\newblock {\em Log. Methods Comput. Sci. 14}, 2 (2018).

\bibitem{BodirskyN06}
{\sc Bodirsky, M., and Nesetril, J.}
\newblock Constraint satisfaction with countable homogeneous templates.
\newblock {\em J. Log. Comput. 16}, 3 (2006), 359--373.

\bibitem{BodirskyP15}
{\sc Bodirsky, M., and Pinsker, M.}
\newblock Schaefer's theorem for graphs.
\newblock {\em J. {ACM} 62}, 3 (2015), 19:1--19:52.

\bibitem{BulatovFVConjecture}
{\sc Bulatov, A.~A.}
\newblock {A dichotomy theorem for nonuniform CSPs}.
\newblock In {\em 2017 IEEE 58th Annual Symposium on Foundations of Computer
  Science (FOCS)\/} (2017), pp.~319--330.

\bibitem{CharatonikW08-bis}
{\sc Charatonik, W., and Wrona, M.}
\newblock Quantified positive temporal constraints.
\newblock In {\em Computer Science Logic, 22nd International Workshop, {CSL}
  2008, 17th Annual Conference of the EACSL, Bertinoro, Italy, September 16-19,
  2008. Proceedings\/} (2008), M.~Kaminski and S.~Martini, Eds., vol.~5213 of
  {\em Lecture Notes in Computer Science}, Springer, pp.~94--108.

\bibitem{CharatonikW08}
{\sc Charatonik, W., and Wrona, M.}
\newblock Tractable quantified constraint satisfaction problems over positive
  temporal templates.
\newblock In {\em Logic for Programming, Artificial Intelligence, and
  Reasoning, 15th International Conference, {LPAR} 2008, Doha, Qatar, November
  22-27, 2008. Proceedings\/} (2008), I.~Cervesato, H.~Veith, and A.~Voronkov,
  Eds., vol.~5330 of {\em Lecture Notes in Computer Science}, Springer,
  pp.~543--557.

\bibitem{HubieSIGACT}
{\sc Chen, H.}
\newblock A rendezvous of logic, complexity, and algebra.
\newblock {\em {SIGACT} News 37}, 4 (2006), 85--114.

\bibitem{Chen09}
{\sc Chen, H.}
\newblock Existentially restricted quantified constraint satisfaction.
\newblock {\em Inf. Comput. 207}, 3 (2009), 369--388.

\bibitem{ChenAAA2019}
{\sc Chen, H.}
\newblock Three concrete complexity questions about constraints, in search of
  theories.
\newblock Invited talk at Arbeitstagung Allgemeine Algebra (Workshop on General
  Algebra), AAA98, Dresden, Germany, July 2019.

\bibitem{ChenMuellerLMCS}
{\sc Chen, H., and M{\"{u}}ller, M.}
\newblock An algebraic preservation theorem for aleph-zero categorical
  quantified constraint satisfaction.
\newblock {\em Log. Methods Comput. Sci. 9}, 1 (2012).

\bibitem{ChenW12}
{\sc Chen, H., and Wrona, M.}
\newblock Guarded ord-horn: {A} tractable fragment of quantified constraint
  satisfaction.
\newblock In {\em 19th International Symposium on Temporal Representation and
  Reasoning, {TIME} 2012, Leicester, United Kingdom, September 12-14, 2012\/}
  (2012), B.~C. Moszkowski, M.~Reynolds, and P.~Terenziani, Eds., {IEEE}
  Computer Society, pp.~99--106.

\bibitem{EiterG95}
{\sc Eiter, T., and Gottlob, G.}
\newblock The complexity of logic-based abduction.
\newblock {\em J. {ACM} 42}, 1 (1995), 3--42.

\bibitem{hodges_1993}
{\sc Hodges, W.}
\newblock {\em Model Theory}.
\newblock Encyclopedia of Mathematics and its Applications. Cambridge
  University Press, 1993.

\bibitem{KrokhinJJ03}
{\sc Krokhin, A.~A., Jeavons, P., and Jonsson, P.}
\newblock Reasoning about temporal relations: The tractable subalgebras of
  allen's interval algebra.
\newblock {\em J. {ACM} 50}, 5 (2003), 591--640.

\bibitem{Wrona14}
{\sc Wrona, M.}
\newblock Tractability frontier for dually-closed ord-horn quantified
  constraint satisfaction problems.
\newblock In {\em Mathematical Foundations of Computer Science 2014 - 39th
  International Symposium, {MFCS} 2014, Budapest, Hungary, August 25-29, 2014.
  Proceedings, Part {I}\/} (2014), E.~Csuhaj{-}Varj{\'{u}}, M.~Dietzfelbinger,
  and Z.~{\'{E}}sik, Eds., vol.~8634 of {\em Lecture Notes in Computer
  Science}, Springer, pp.~535--546.

\bibitem{ZhukFVConjecture}
{\sc {Zhuk}, D.}
\newblock A proof of {CSP} dichotomy conjecture.
\newblock In {\em 2017 IEEE 58th Annual Symposium on Foundations of Computer
  Science (FOCS)\/} (Oct 2017), pp.~331--342.

\bibitem{Zhuk20}
{\sc Zhuk, D.}
\newblock A proof of the {CSP} dichotomy conjecture.
\newblock {\em J. {ACM} 67}, 5 (2020), 30:1--30:78.

\bibitem{abs-2110-09504}
{\sc Zhuk, D.}
\newblock The complexity of the quantified {CSP} having the polynomially
  generated powers property.
\newblock {\em CoRR abs/2110.09504\/} (2021).

\bibitem{zhuk2020qcsp}
{\sc Zhuk, D., and Martin, B.}
\newblock {QCSP} monsters and the demise of the chen conjecture.
\newblock In {\em Proceedings of the 52nd Annual ACM SIGACT Symposium on Theory
  of Computing\/} (2020), pp.~91--104.

\end{thebibliography}
\end{document}